\documentclass{tran-l}

\usepackage[colorlinks,citecolor=blue,urlcolor=blue]{hyperref}
\usepackage{graphicx}
\usepackage[shortlabels]{enumitem}
\usepackage{tikz}
\usepackage{float}
\usepackage{subcaption}
\usepackage{mathptmx}
\usepackage[T1]{fontenc}
\usepackage[numbers]{natbib}
\usepackage{lineno}
\usepackage{amsthm,amsmath,amsfonts,amssymb}

\newtheorem{theorem}{Theorem}[section]
\newtheorem{lemma}[theorem]{Lemma}

\theoremstyle{definition}

\theoremstyle{remark}
\newtheorem{remark}[theorem]{Remark}

\numberwithin{equation}{section}

\begin{document}

% \title[short text for running head]{full title}
\title[Spectral Analysis of Nonstationary Anderson Model Operators]{Spectral Analysis of Lattice Schrödinger-Type Operators Associated with the Nonstationary Anderson Model and Intermittency}

%    Only \author and \address are required; other information is
%    optional.  Remove any unused author tags.

%    author one information
% \author[short version for running head]{name for top of paper}

\author{Dan Han}
\address{Department of Mathematics,
University of Louisville}
\email{dan.han@louisville.edu}

%    author two information
\author{Stanislav Molchanov}
\address{Department of Mathematics and Statistics, University of North Carolina at Charlotte}
\email{smolchan@charlotte.edu}

%    author three information
\author{Boris Vainberg}
\address{Department of Mathematics and Statistics, University of North Carolina at Charlotte}
\email{brvainbe@charlotte.edu}

%    \subjclass is required.
\subjclass[2020]{60H25;60H15;81Q10;37H15;35B40}

\date{}

\dedicatory{}

%    Abstract is required.
\begin{abstract}
The research explores a high irregularity, commonly referred to as intermittency, of the solution to the non-stationary parabolic Anderson problem:
\begin{equation*}
  \frac{\partial u}{\partial t} = \varkappa \mathcal{L}u(t,x) + \xi_{t}(x)u(t,x)
\end{equation*}
with the initial condition \(u(0,x) \equiv 1\), where \((t,x) \in [0,\infty)\times \mathbb{Z}^d\). Here, \(\varkappa \mathcal{L}\) denotes a non-local Laplacian, and \(\xi_{t}(x)\) is a correlated white noise potential. The observed irregularity is intricately linked to the upper part of the spectrum of the multiparticle Schr\"{o}dinger equations for the moment functions \(m_p(t,x_1,x_2,\cdots,x_p) = \langle u(t,x_1)u(t,x_2)\cdots u(t,x_p)\rangle\).

In the first half of the paper, a weak form of intermittency is expressed through moment functions of order $p\geq 3$ and established for a wide class of operators $\varkappa \mathcal{L}$ with a positive-definite correlator $B=B(x))$ of the white noise. In the second half of the paper, the strong intermittency is studied. It relates to the existence of a positive eigenvalue for the lattice Schrödinger type operator with the potential $B$.
This operator is associated with the second moment $m_2$.  Now $B$ is not necessarily positive-definite, but $\sum B(x)\geq 0$.
\end{abstract}

\maketitle

%    Text of article.

\section{Introduction}\label{sec:introduction}

This paper presents novel advancements in the study of nonstationary lattice Anderson parabolic problems, building upon two recent works \cite{dan2023non} \cite{molchanov2022positive} and memoir \cite{carmona1994parabolic}. The central object of \cite{dan2023non} is the nonstationary lattice Anderson parabolic problem
\begin{equation}\label{andersonmodel}
  \begin{aligned}
      \frac{\partial u}{\partial t} &=\varkappa \mathcal{L}u(t,x)+\xi_{t}(x)u(t,x)\\
       u(0,x) &\equiv 1, \ (t,x) \in [0,\infty)\times Z^d.
  \end{aligned}
\end{equation}
with non-local Laplacian 
\begin{equation}\label{nonlocal Laplacian}
\mathcal{L}f(t,x)=\sum\limits_{\substack{z\neq 0\\ z\in Z^d}}a(z)(f(t,x+z)-f(t,x)).
\end{equation}

The operator $\mathcal{L}$ is the generator of the random walk $x(t)$. The random walk $x(t)$ has the following structure. It spends in each site $x \in Z^d$ the time $\tau_x$ which is exponential distributed with coefficient $\varkappa$, i.e., $P(\tau_x >t) =  e^{-\varkappa t}$, and it jumps at the moment $\tau_x+0$ from site $x$ to site $x+z$ with probability $a(z)$. We assume $a(z)=a(-z)\geq 0$ (symmetry), $\sum\limits_{ z \in Z^d }a(z)=1$ (normalization), $a(0)=0$, $a(z)>0$ if $|z|=1$ (to guarantee non-degeneracy of the random walk). The constant $\varkappa$ in (\ref{andersonmodel}) is the diffusion coefficient (diffusivity).

The potential $\xi_t(x)$, $x\in Z^d$ is the Gaussian field which is $\delta_0$-correlated in time and correlated in space:
\begin{align}\label{wn}
<\xi_{t_1}(x_1)\xi_{t_2}(x_2)>=\delta_0(t_1-t_2)B(x_1-x_2)
\end{align}
where $\delta_0(t)$ is the standard Dirac delta function. It means that $\xi_t(x)$ is the white noise in time with space correlator $B(x_1-x_2), ~x_1,x_2\in Z^d$.

Our goal is to qualitatively analyze the field $u(t,x)$ for $t\to\infty, x\in Z^d$, and prove that $u$ is intermittent in space. At the physical level, the intermittency means that the main contribution to the moment functions of $u(t,x),~t\to\infty,$ is given by very sparse and very high picks \cite{foondun2009intermittence,zel1987intermittency}. The mathematical description and study of the intermittency are using Lyapunov exponents, properties of the correlation function $B$, and the multi-particle equation for higher order moments of $u$.

{\it Properties of the potential $\xi_t(x)$ and the correlation function $B$:} We can understand $\xi_t(x)$ as the generalized derivative of the Wiener process $W(t,x),t\geq 0,x\in Z^d$, That is $\xi_t(x)=\dot{W}(t,x)$ and equation (\ref{andersonmodel}) is understood as the infinite dimensional stochastic differential equations in It\'{o} sense in an appropriate weighted Hilbert space $L^2(Z^d,\mu)$ with measure $\mu$. The measure $\mu$ on $Z^d$ is selected based on the jumps distribution $a(z)$ $z\in Z^d$, see details in \cite{dan2023non}. We will use the notation $<\cdot>$ to represent the integration over the distribution of $\xi_t(x)=\dot{W}(t,x)$, and $E[\cdot]$ will be the expectation of the functional of the random walk $x(t)$ with the generator $\varkappa\mathcal{L}$. Formula (\ref{wn}) implies that $<W(t,x)>=0$ and $<W(t_1,x_1)W(t_2,x_2)>=\min(t_1,t_2)B(x_1-x_2)$.
We construct the correlated Wiener process $W(t,x)$ as a linear transform of  independent and identically distributed (i.i.d) standard Brownian motions $w(t,x)$, $x \in Z^d$ , $t\geq 0$ with weight kernel $b(\cdot)$:
\begin{align}
W(t,x) = \sum\limits_{z \in Z^d} b(x-z) w(t,z).
\end{align}
Assume that the kernel $b(\cdot)$ satisfies  $b(\cdot) \in l^1(Z^d)$, which implies $\sum\limits_{z\in Z^d}|b^2(z)|<\infty$ and
\begin{align}
B(x-y) &= < W^2(1,x-y) > = \sum\limits_{z\in Z^d}b(x-z)b(y-z)=\sum\limits_{z\in Z^d}b(x-y-z)b(-z)\label{B(x-y)about b1}\\
B(0) &= < W^2(1,0) > = \sum_{z \in Z^d} b^2(z)<\infty.\label{B(x-y)about b2}
\end{align}

The function $B(x-y)$ is the space correlation of stationary Gaussian field $W(t,x)$. Due to Bochner-Khinchin theorem, $B$ is a positively definite function, i.e., it has the representation
\begin{align}\label{b3}
B(x)=\frac{1}{(2\pi)^d}\int_{T^d} e^{ikx}\hat{B}(k)dk,  \quad T=[-\pi,\pi]^d,
\end{align}
where the spectral density $\hat{B}(k)=\sum\limits_{x\in Z^d}B(x)e^{ikx}\in C(T^d)$ and $\hat{B}(k)\geq0$. Indeed, since $b(x)\in l^1(Z^d)\bigcap l^2(Z^d)$, formulas (\ref{B(x-y)about b1}), (\ref{B(x-y)about b2}) imply that $\hat{B}(k)=|\hat{b}(k)|^2\geq 0$ where $\hat{b}(k)=\sum\limits_{x\in Z^d}e^{ikx}b(x),~k\in T^d$.

Note also that
\begin{align*}
\sum\limits_{x\in Z^d}B(x)=\frac{1}{(2\pi)^d}\int_{T^d}\sum\limits_{x\in Z^d}e^{ikx}\hat{B}(k)dk=\int_{T^d}\delta_0(x)\hat{B}(k)dk=\hat{B}(0)\geq 0.
\end{align*}

{\it Higher order moments and multi-particle equation:} Equation (\ref{andersonmodel}) can be rewritten in the integral form (in It\'{o} sense):
\begin{align}\label{ito equation}
u(t,x)=1+\int_0^t\varkappa\mathcal{L}u(s,x)ds+\int_0^tu(s,x)dW(s,x).
\end{align}
The existence and uniqueness theorems for stochastic differential equation (\ref{ito equation}) are given in the paper \cite{dan2023non}. One also can find there 
%The memoirs \cite{carmona1994parabolic} contains the detailed analysis of the important special case when $\mathcal{L}$ is the local lattice Laplacian:
%\begin{equation*}
%\mathcal{L}\psi(x)=\Delta\psi(x)=\frac{1}{(2\pi)^d}\sum\limits_{x':|x'-x|=1}(\psi(x')-\psi(x)),
%\end{equation*}
 %and $B(x-y)=\delta(x-y)$. The white noise $\xi_t(x)$, $x\in Z^d$ are independent for different $x\in Z^d$. The paper \cite{dan2023non} is the generalization of the work in \cite{carmona1994parabolic} in several directions.
  the derivations of the equations for the $p$-th moment functions $m_p(t,x_1,x_2,\cdots,x_p)=<u(t,x_1)\cdots u(t,x_p)>, ~p=1,2,\cdots$:
\begin{equation} \label{eq: p moment}
	\frac{d m_p(t, x_1, \cdots, x_p) }{dt} = \varkappa \left(\sum_{i=1}^p \mathcal{L}_{x_i} m_p(t, x_1, \cdots, x_p) \right) +V_p(x_1,\cdots,x_p) m_p(t, x_1, \cdots, x_p)=H_pm_p
\end{equation}
with the initial condition $m_p(0, x_1, \cdots, x_p) =1$. Here $ V_p(x_1,\cdots,x_p)= \sum\limits_{i <j} B(x_i-x_j)$.

The first two moments of the field $u(t,x)$ satisfy the following equations respectively:
\begin{align}
\frac{\partial m_1}{\partial t}=(\varkappa\mathcal{L})m_1(t,x), \quad ~
m_1(0,x)=1,
\end{align}
i.e., $m_1(t,x)=1$. If $m_2(t,x)=m_2(t,x+z,z)$, then
\begin{align}\label{2ndmoment_L}
\frac{\partial m_2(t,x)}{\partial t}=2\varkappa \mathcal{L}m_2(t,x)+B(x)m_2(t,x)=H_2m_2, \quad ~m_2(0,x)=1.
\end{align}

For the special case when $\varkappa\mathcal{L}$ is the local Laplacian operator $\varkappa\Delta$ and $B(x_1-x_2)=\delta(x_1-x_2)$, the derivations and analysis of the equations for the $p$th moment functions $m_p(t,x_1,x_2,\cdots,x_p)=<u(t,x_1)\cdots u(t,x_p)>, ~p=1,2,\cdots$ can be found in \cite{carmona1994parabolic}. In the latter case, there exists the exact formula for $m_2(t,x)$.

All equations for higher order moments have the form of multi-particle Schr\"{o}dinger equations containing the Hamiltonian $H_p$ as in the equation (\ref{eq: p moment}). 
%Let us look more carefully at $H_2$. 
The ``non-local Laplacian" $\varkappa\mathcal{L}$ is a symmetric bounded and self-adjoint operator in $l^2(Z^d)$ with the dot product $(f,g)=\sum\limits_{z\in Z^d}f(x)\bar{g}(x)$. It is isomorphic to the operator of multiplication by the continuous non-positive functions, the Fourier symbol, $\varkappa\hat{\mathcal{L}}(k)=\varkappa\sum\limits_{z\in Z^d}(\cos(kz)-1)a(z)=\varkappa(\hat{a}(k)-1),~k \in T^d=[-\pi,\pi]^d$, in the dual Fourier space $L^2(T^d, dk)$.
As a result, the spectrum $Sp(\varkappa\mathcal{L})$ of $\varkappa\mathcal{L}$ coincides with $Range(\varkappa\hat{\mathcal{L}}(k))=[\varkappa \alpha,0], ~\alpha=\min\limits_{k\in T^d}\hat{\mathcal{L}}(k)=\min\limits_{k\in T^d}(\hat{a}(k)-1)<0$.

{\it Spectral analysis of $H_p$ and intermittency:} The notion of intermittency presented in the physical literature refers to a high irregularity of the fields, see \cite{avellaneda1992renormalization,foondun2009intermittence,she1994universal,zel1987intermittency}.  In the case of the solar magnetic field, the intermittency pertains to the presence of high and narrow peaks that encapsulate nearly all magnetic energy. In mathematical literature \cite{dan2023non,gaudreau2023moment,getan2017intermittency,greven1991population,koralov2013structure,lyu2023almost,zel1987intermittency}, this notion is described in terms of the Lyapunov exponents $\gamma_p$ for the moments $m_p(t,\dot)$ defined in (\ref{eq: p moment}). Estimates on $\gamma_p$  depend on the structure of the spectral measure for the Hamiltonian $H_p$, $p\geq 2$, near the top of the spectrum. 

Due to the stationarity of the field $u(t,x),x\in Z^d, $ the moments do not depend on $x$: $m_p(t)=<u^p(t,x)>=<u^p(t,0)>$. We will prove (see section 2) that the following Lyapunov exponents exist:
\begin{align*}
\gamma_1=\lim\limits_{t\rightarrow \infty} \frac{\ln <u(t,0)>}{t}=\lim\limits_{t\rightarrow \infty}\frac{\ln m_1(t)}{t}=0\\
\gamma_p=\lim\limits_{t\rightarrow \infty}\frac{\ln<u^p(t,0)}{t}=\lim\limits_{t\rightarrow\infty}\frac{\ln m_p(t)}{t}, p\geq 2,
\end{align*}
and that $0= \gamma_1\leq \frac{\gamma_2}{2}\leq \frac{\gamma_3}{3}\leq \cdots$ (Lyapunov type inequality). If for some index $\rho_0$, we have
\begin{align*}
0=\gamma_1=\frac{\gamma_2}{2}=\cdots=\frac{\gamma_{p_0}}{p_0}<\frac{\gamma_{p_0+1}}{p_0+1},
\end{align*}
then for $k\geq p_0+1$, we have $\frac{\gamma_k}{k}<\frac{\gamma_{k+1}}{k+1}$ and this fact is the manifestation of intermittency \cite{den2012random,zel1987intermittency}. If $\frac{\gamma_2}{2}>0$, i.e., $p_0=1$, then \textit{the intermittency is strong}; if $p_0>1$, \textit{the intermittency is weak}. The qualitative meaning of intermittency will be explained in the remark in section \ref{sec:intermittency}.

The main part of the paper contains two sections. In Section \ref{sec:intermittency}, we establish the weak intermittency of the solution $u(t,x)$. This is deduced from simple upper and lower estimations of the moments $m_p(t)$ for $p=2,3,\cdots$, that imply the following estimates on the Lyapunov exponents $\gamma_p$:

\begin{align}
\frac{p-1}{2}B(0) - \varkappa d \leq \frac{\gamma_p}{p} \leq \frac{(p-1)B(0)}{2}.
\end{align}

The second part of Section \ref{sec:intermittency} contains improved upper and lower estimates for $\frac{\gamma_p}{p},~p\geq 2$ in terms of $\gamma_2$. These results reveal that the critical value $p_0=p_0(\varkappa)$, defined as the minimal $p_0$ such that $0<\frac{\gamma_{p_0}}{p_0}$, has order $ O(\varkappa)$ for large $\varkappa$.

The final section \ref{sec:spectral analysis} focuses on the examination of the second moment and the second Lyapunov exponent $\gamma_2$. In the transient case, under minimal conditions on the generator $\varkappa\mathcal{L}$ and the correlator $B(x)$, we establish the existence of a phase transition:  $\gamma_2(\varkappa)=0$ if $\varkappa<\varkappa_0$; $\gamma_2(\varkappa)>0$ if $\varkappa>\varkappa_0$ for appropriate $\varkappa_0$ if $\sum\limits_{x\in Z^d}B(x) > 0$. In the recurrent case, we consider a very broad class of operators and show that $\gamma_2(\varkappa)>0$ for any $\varkappa>0$. The positivity of $\gamma_2(\varkappa)$ is established for an arbitrary, not necessarily Markovian pseudodifferential operator $\varkappa \mathcal{L}$. We do not require the potential $B$ to be positive-definite.  The main assumption on the potential is that $\sum\limits_{x\in Z^d}B(x) > 0$. In a specific case, we demonstrate that $\gamma_2(\varkappa) > 0$ for $\varkappa > 0$, even in the borderline scenario where $\sum\limits_{x\in Z^d}B(x) = 0$.

\section{Intermittency of the Anderson parabolic problem}\label{sec:intermittency} 
The intermittency phenomenon is associated with the progressive growth of moment functions $m_p(t,x_1,x_2,\cdots,x_d)=<u(t,x_1)\cdots u(t,x_p)>$ as $t$ tends to infinity. Equations (\ref{eq: p moment}) for $m_p(t,x_1,\cdots,x_d)$ have been derived 
%for arbitrary generators $\varkappa \mathcal{L}$ and correlated white noise $\xi_t(x)$ 
in \citep{dan2023non}, but the analysis of the equation has been limited to the cases where $p=1,2$. Here we extend the study by providing several estimations for $m_p(\cdot,\cdot)$, specifically for $p\geq 3$, shedding light on the phenomenon of intermittency.

The existence-uniqueness theorem for stochastic differential equations governing $u(t,x)$ in a weighted Hilbert space, along with the Itô formula, supports the derivation of equations (\ref{eq: p moment}). The solution $u(t,x)$ of the stochastic differential equation (\ref{andersonmodel}) is homogeneous and ergodic in space, and all its moment functions $m_p(t,x_1,x_2,\cdots,x_p)=<u(t,x_1)u(t,x_2)\cdots u(t,x_p)>$ are finite. Furthermore, these moment functions satisfy the $p$-particle type Schr\"{o}dinger equation (\ref{eq: p moment})
Since $\|\varkappa\mathcal{L}\|\leq \varkappa$ and $\max\limits_{x\in Z^d}|B(x)|=\max\limits_{x\in Z^d}B(x)=B(0)$, i.e., $|V_p|\leq \frac{p(p-1)}{2}B(0)$, the Hamiltonian $H_p$ in  (\ref{eq: p moment}) is bounded for any $p\geq 1$. 

Note also that since $u(t,x)>0$,
\begin{align} \label{m0}
m_p(t,x_1,x_2,\cdots,x_p)&=<u(t,x_1)u(t,x_2)\cdots u(t,x_p)>\\ \nonumber
&\leq \frac{1}{p}<u^p(t,x_1)+\cdots+u^p(t,x_p)>\\
\nonumber &=\frac{1}{p}\sum\limits_{j=1}^p<u^p(t,x_j)>=m_p(t,0,0,\cdots,0)
\end{align}

Due to Kac-Feynman formula,
\begin{align}\label{Kac Feynman for p particle}
m_p(t,0,\cdots,0)=E_{(0,0,\cdots,0)}e^{\int_0^t V_p(x_1(s),x_2(s),\cdots,x_p(s))ds}
\end{align}
where $x_j(s), j=1,2,\cdots,p,$ are $p$ independent random walks on $Z^d$ with generator $\varkappa\mathcal{L}$. Formula (\ref{Kac Feynman for p particle}) doesn't contain the pre-exponent $u(0,x_1(t),\cdots,x_p(t))$ due to the condition $u(0,x_1,x_2,\cdots,x_p)=1$. Future estimates are based on the formula (\ref{Kac Feynman for p particle}). 

Denote $M_p(t)=m_p(t,0,0,\cdots,0)=\max \limits_{x_1,x_2,\cdots,x_p} m_p(t,x_1,x_2,\cdots,x_p)$, see inequality (\ref{m0}). The max principle for the parabolic equation (\ref{eq: p moment}) and the independence of increments in time for the processes $x_j(t)$ lead to the inequality $M_p(t+s)\leq M_p(t)M_p(s)$, where $M_p\geq 1$. The subadditivity of $\ln M_p(t)$  implies the existence of the Lyapunov exponent, given by
\begin{align*}
\gamma_p=\lim\limits_{t\rightarrow \infty}\frac{\ln M_p(t)}{t}, \quad p=1,2,\ldots.
\end{align*}

It is known that $\gamma_p$ is the upper bound of the spectrum of the multiparticle Hamiltonian $H_p$, $p\geq 2$, see \cite{carmona1994parabolic}. Due to the Lyapunov inequality,

\begin{align}
0=\gamma_1\leq\frac{\gamma_2}{2}\leq \frac{\gamma_3}{3}\leq \cdots\leq \frac{\gamma_p}{p}\leq \cdots
\end{align}
Our goal now is to prove that the field $u(t,x)$ is intermittent in the sense of formal definitions from \cite{carmona1994parabolic,zel1987intermittency,den2012random,koralov2013structure}.

Let us revisit the concept of intermittency, briefly touched upon in the Introduction. 
%This concept, widely acknowledged in the physical literature, is commonly associated with Ya. B. Zel'Dovich \cite{zel1987intermittency}. In our formal mathematical framework, as defined in \cite{carmona1994parabolic} and \cite{den2012random}, intermittency is characterized by the following inequalities:
%\begin{align}
%0 \equiv \gamma_1 < \frac{\gamma_2}{2}(\varkappa) \quad & \text{(Strong intermittency)} \\
%0 \equiv \gamma_1 < \frac{\gamma_{p_0}}{p_0}, \quad p_0 > 2 \quad & \text{(Weak intermittency)}
%\end{align}
In the following remark, we will try to explain why the formal mathematical definition of intermittency implies the high irregularity of the field $u(t,x)$. We will also stress the similarity between our scaler model and models for the Solar magnetic field \cite{den2012random,molchanov1994lyapunov,ruzmaikin1993random,yang2010intermittency}. 

\begin{remark}
The field $u(t,x)$, representing solutions to the Anderson parabolic problem, is both homogeneous and ergodic. The correlation functions of any order are finite. Specifically, it leads to the following ergodic theorem,
\begin{align}
<u^{p_0}(t,x)>=<u^{p_0}(t,0)>=\lim\limits_{L\rightarrow\infty}\frac{1}{(2L)^d}\sum\limits_{|y|\leq L}u^{p_0}(t,y)\stackrel{\log}{\sim}e^{\gamma_{p_0}t}, \quad t\to \infty.
\end{align}

Let $0<\delta<\gamma_{p_0}/p_0$. Chebyshev's inequality and the relation $<u(t,y)>\equiv 1,~y\in Z^d,$ imply
\begin{align*}
P(\frac{\ln u(t,y)}{t}\geq \frac{\gamma_{p_0}}{p_0}-\delta)=P(u(t,y)\geq  e^{t( \frac{\gamma_{p_0}}{p_0}-\delta)})\leq <u(t,y)>e^{-t(\frac{\gamma_{p_0}}{p_0}-\delta)}= e^{-\delta_1t},~~ \delta_1>0.
\end{align*}

For a fixed $\delta$ and a large $L$, consider the cube $Q_{L} = \{y:|y| < L\}$ which we can partition into the set $Q_{L}^{\delta} = \{y \in Q_{L}: u(t,y) \leq e^{-t(\frac{\gamma_{p_0}}{p_0}-\delta)}\}$ and the complement set containing large values of $u$. In a logarithmic sense, we have $\frac{1}{(2L)^d}\sum\limits_{Q_L^{\delta}}<u^{p_0}(t,x)>\leq e^{t(\gamma_{p_0}-p_0\delta)}$, which is significantly smaller than the logarithmic asymptotic $e^{\gamma_{p_0}t}$ of $<u^{p_0}(t,0)>$. This indicates that the primary contribution to $<u^{p_0}(t,0)>$ arises from the sparse high peaks where $u(t,x)>e^{-t(\frac{\gamma_{p_0}}{p_0}-\delta)}$. The density $\pi(t,\delta)=|Q_L\setminus Q_L^\delta|/|Q_L|$ of such points in $Q_L$, $L\rightarrow \infty$,  is exponentially small in $t$:
\begin{align*}
\pi(t,\delta)\leq e^{-t\delta_1}.
\end{align*}

A classic example of \textit{intermittent field} is the magnetic field of the Sun and a diverse array of white and yellow stars. The energy of magnetic field $\vec{H}(t,x)$ is proportional to the 2nd moment $<\vec{H}(t,x_1)\vec{H}(t,x_2)>$. The Lyapunov exponents of $\vec{H}(x)$ satisfy the same relation $0=\gamma_1<\frac{\gamma_2}{2}$ for hot stars including the Sun. 
Nearly all magnetic energy of $\vec{H}(x)$ is concentrated within the black spots, covering extremely small regions of the Sun's surface, as discussed in \cite{zel1987intermittency}. The lattice parabolic Anderson model, introduced in \cite{carmona1994parabolic}, serves as a simplified representation of the Maxwell equations governing the magnetic field in the turbulent flow of solar plasma.
    
\end{remark}

\begin{theorem}
The field $u(t,x)$, representing the solution to the equation (\ref{andersonmodel}) as an infinite-dimensional stochastic differential equation (SDE) in the appropriate weighted Hilbert space, is \textit{weak intermittent} as $t\rightarrow\infty$. This holds without any additional restrictions on the jump distribution $a(x),~x\in Z^d$, and the correlator $B(\cdot)$, except as outlined in the introduction.
\end{theorem}

\begin{proof}
Due to (\ref{m0}) and the maximum principle, 
\begin{align*}
M_p(t)=\max\limits_{x_1,x_2,\cdots,x_p}m_p(t,x_1,x_2,\cdots,x_p)=M_p(t,0,\cdot 0)\leq \exp\{\frac{p(p-1)}{2}B(0)t\},
\end{align*}

i.e.,
\begin{align}\label{gamma_p/p upper bound}
\frac{\gamma_p}{p}\leq \frac{p-1}{2}B(0), p=2,3,\cdots~.
\end{align}
At the same time,
\begin{align*}
M_p(t,0,\cdots,0)\geq E_{0,0,\cdots,0}[I_{\{y_1(s)=y_2(s)=\cdots=y_p(s)=0,s\in [0,t]\}}]e^{\frac{p(p-1)}{2}B(0)t}
\end{align*}
\begin{align*}
=\exp\{-p\varkappa dt+\frac{p(p-1)}{2}B(0)t\},
\end{align*}
i.e., 
\begin{align}\label{gamma_p/p lower bound}
\frac{\gamma_p}{p}\geq \frac{p-1}{2}B(0)-\varkappa d.
\end{align}

Thus,
\begin{align}
\frac{\gamma_p}{p}>0 \quad \,\,\text{if} \quad \,\,p>2+\frac{2\varkappa}{B(0)}.
\end{align}

This implies the weak intermittency of the field $u(t,x)$ and provides an estimate from below on the borderline value of $p=p_0$.

\end{proof}

Estimates (\ref{gamma_p/p lower bound}) and (\ref{gamma_p/p upper bound}) are very rough.  In the following, we will present significantly refined estimates for $\gamma_p$ with $p \geq 2$. The pivotal element in this refinement lies in the accurate estimation of the second moment $M_2(t)$ and the corresponding Lyapunov exponent $\gamma_2$.

\subsection{Upper estimate of Lyapunov exponents $\gamma_p$ through $\gamma_2$}
We will start with two technical lemmas. Let us introduce a little more general stochastic differential equation (\ref{ito equation}) where the white noise has an extra factor $\alpha$. Denote by $\tilde{u}(t,x),~\tilde{m}_p(t, \vec{x}),~\tilde{\gamma}_p$ the corresponding solution of (\ref{ito equation}) with factor $\alpha$, its moment functions and Lyapunov exponents. Here $ \vec{x}=(x_1,\cdots,x_p)$. In particular, function $\tilde{m}_p(t, \vec{x})$ satisfies 
\begin{align}\label{eq:alternative tilde m_p}
    \frac{d \tilde{m}_p(t, \vec{x}) }{dt} &= \varkappa \left(\sum_{i=1}^p \mathcal{L}_{x_i} \tilde{m}_p(t, \vec{x}) \right) +\alpha V_p(\vec{x}) \tilde{m}_p(t, \vec{x})=\tilde{H}_p\tilde{m}_p\\
    \tilde{m}_p(0, \vec{x}) &=1.
\end{align}

The solution to equation (\ref{eq:alternative tilde m_p}) is given by
\begin{align}\label{eq:alternative tilde m_p Kac Feynman}
\tilde{m}_p(t,\vec{x}) = E_{\vec{x}} \exp\left(\alpha\int_{0}^{t} V_p(t,\vec{x}_s)ds\right)
\end{align}
where $\vec{x}_s=(x_1(s),x_2(s),\cdots,x_p(s))$ is the family of independently random walks with the generator $\varkappa\mathcal{L}$. 
We will explicitly show the dependence of the Lyapunov exponent on two parameters: $\tilde{\gamma}_p=\tilde{\gamma}_p(\varkappa,\alpha)$.

\begin{lemma}\label{lem:two gamma relationship}
Under the aforementioned conditions, $\tilde{\gamma}_p(\varkappa,\alpha) = \alpha\gamma_p\left(\displaystyle\frac{\varkappa}{\alpha}\right)$.
\end{lemma}
\begin{proof}
    This result is derived through direct calculations. We rescale time by setting $\alpha s = \tau$ in (\ref{eq:alternative tilde m_p Kac Feynman}) and obtain:
\begin{align}
\tilde{m}_p(t,\vec{x}) = E_{\vec{x}} \exp\left(\int_{0}^{\alpha t} V_p(t,\vec{x}_{\tau/\alpha}) d\tau \right)
\end{align}
The new process $\vec{y}(t)=\vec{x}_{\frac{\tau}{\alpha}}$ has the generator $\frac{\varkappa}{\alpha}\mathcal{L}$, and 
\begin{align}
    \tilde{\gamma}(\varkappa,\alpha)=\lim\limits_{t\rightarrow\infty}\frac{\ln \tilde{m}_p(t,\boldsymbol{0})}{t}=\lim\limits_{t \rightarrow \infty}\frac{\alpha \ln m_p(\tau,\boldsymbol{0})}{\tau}=\alpha\gamma_p(\frac{\varkappa}{\alpha})
\end{align}
\end{proof}
The following lemma gives a special representation of the multiparticle potential 
\begin{align}\label{Vp}
V_p(x_1, x_2, \ldots, x_p) = \sum_{1 \leq i < j \leq p} B(x_i - x_j).
\end{align}
\begin{lemma}
If $p = 2l+1$ is odd and the number of terms in $V_p$ is ${p \choose 2} = l(2l+1)$, then there exists a partition of the sum (\ref{Vp}) in $p$ groups, $V_p = G_1 + \cdots + G_p$, such that each group contains $l$ terms $B(x_i - x_j)$ with the following property: the group $G_k,~1\leq k\leq p,$ does not depend on $x_k$, and each other argument $x_s$ in the terms $B(x_i - x_j)$ inside of the group $G_k$ occurs only once. 

If $p = 2l$ is even and the number of terms in $V_p$ is ${p \choose 2} = l(2l-1)$, then there exists a partition of the sum (\ref{Vp}) in $p-1$ groups, $V_p = G_1 + \cdots + G_{p-1}$, such that each group contains $l$ terms $B(x_i - x_j)$ with the following property: each argument $x_s$ in the terms $B(x_i - x_j)$ inside of a group occurs only once.

\end{lemma}
\begin{figure}[h]
\includegraphics[scale=0.25]{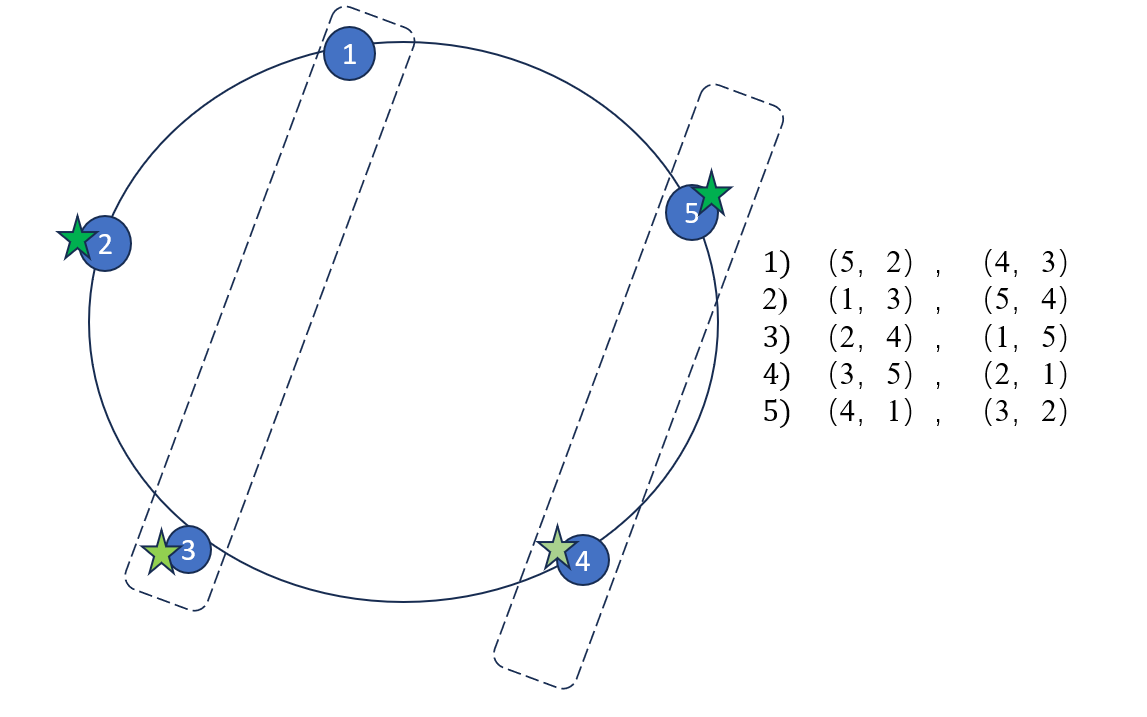}
\caption{An example of $G_2$ when $p=5,l=2$.}
\label{fig:p5l3_example}
\end{figure}

The Figure (\ref{fig:p5l3_example}) depicts the algorithmic process for constructing a configuration, with specific values assigned for $l$ and $p$ (namely, $l=2$ and $p=5$). It is crucial to note that this construction is not unique. Along the unit circle, we identify $p=2l+1$ points, numbered clockwise as $x_1, x_2, \ldots, x_{2l+1}$. To elaborate, we fix $1 \leq i \leq p$, exclude $x_i$, and subsequently generate successive pairs using the remaining variables. This is achieved by moving clockwise and counterclockwise along the equal distribution of points: $(x_{i+1},x_{i-1}),(x_{i+1},x_{i-1}),\cdots$. Then 
\begin{align*}
V(x_1,x_2,x_3,x_4,x_5)=&\underbrace{B(x_5-x_2)+B(x_4-x_3)}_{G_1}+\underbrace{B(x_1-x_3)+B(x_5-x_4)}_{G_2}\\
&+\underbrace{B(x_2-x_4)+B(x_1-x_5)}_{G_3}+\underbrace{B(x_3-x_5)+B(x_2-x_1)}_{G_4}\\
&+\underbrace{B(x_4-x_1)+B(x_3-x_2)}_{G_5}.
\end{align*}

\begin{proof}
When $p=2l+1$, the first part directly follows from our specific partition, as depicted in Fig. (\ref{fig:p5l3_example}). For $p=2l$, where $l>1$, we proceed by isolating $x_p$ and applying the partition to the remaining $2l-1$ points. We divide them into $p-1=2l-1$ groups $G_i$, where $i=1,2,\cdots,p-1$, with each $G_i$ containing $l-1$ possible pairs among $\{x_1,x_2,\cdots,x_{p-1}\}$-$\{x_i\}$. Next, we add $B(x_p-x_i)$ to each group $G_i$. For example, when $p=6$ and $l=3$,
\begin{align*}
V&=\underbrace{B(x_5-x_2)+B(x_4-x_3)+B(x_6-x_1)}_{G_1}+\underbrace{B(x_1-x_3)+B(x_5-x_4)+B(x_6-x_2)}_{G_2}\\
&+\underbrace{B(x_4-x_2)+B(x_5-x_1)+B(x_6-x_3)}_{G_3}+\underbrace{B(x_3-x_5)+B(x_2-x_1)+B(x_6-x_4)}_{G_4}\\
&+\underbrace{B(x_4-x_1)+B(x_3-x_2)+B(x_6-x_5)}_{G_5}.
\end{align*}

This method of construction ensures a systematic and comprehensive generation of pairs, as exemplified for the given parameters.
\end{proof}

The Theorem (\ref{thm:upper estimate of Lyapunov exponent}) highlights the efficiency of direct probabilistic methods within the pure spectral context. Providing a functional analytical proof for the same result is non-trivial. Estimations of the higher moments $M_p(t)=<u^p(t,0,0,\cdots,0)>=<E\exp(\sum\limits_{1\leq i<j\leq p}\int_0^t B(x_i(s)-x_j(s)ds>$ includes the additional improvement based on the fact that $(x_1(s)-x_2(s)),(x_3(s)-x_4(s)),\cdots$ etc are independent processes.  
\begin{theorem}\label{thm:upper estimate of Lyapunov exponent}
\textbf{Upper estimate of Lyapunov exponents} For $p\geq 3$
\begin{align}
    \frac{\gamma_p}{p}=\frac{1}{p}\lim\limits_{t\rightarrow \infty}\frac{\ln M_p(t)}{t}\leq [\frac{p}2]\gamma_2(\frac{\varkappa}{p})
\end{align}
\end{theorem}

\begin{proof}

Consider $M_p(t)=E_{(0,0,\cdots,0)}[\exp(\sum\limits_{i<j}\int_0^tB(x_i(s)-x_j(s)ds)]$, using the inequality $$a_1a_2\cdots a_p\leq \frac{1}{p}(a_1^p+\cdots a_p^p),$$ we can estimate moment function $M_p(t)$ from the above.

\begin{align*}
M_p(t)&=E_{(0,\cdots,0)}[\exp(\sum\limits_{1\leq i<j\leq p}(\int_0^t B(x_i(s)-x_j(s))ds))]\\
&= E_{(0,\cdots,0)}[\exp(\sum\limits_{l=1}^p\sum\limits_{(i,j)\in G_l}(\int_0^tB(x_i(s)-x_j(s))ds))]\\
&\leq \frac{1}{p}\sum\limits_{l=1}^pE_{(0,\cdots,0)}[\exp(p\sum\limits_{(i,j)\in G_l}(\int_0^tB(x_i(s)-x_j(s))ds))]\\
&= E_{(0,\cdots,0)}[\exp(p\sum\limits_{(i,j)\in G_l}(\int_0^tB(x_i(s)-x_j(s))ds))]\\
& =\{E_{(0,\cdots,0)}[\exp(p(\int_0^tB(x_1(s)-x_2(s))ds))]\}^l,\,\, l=[\frac{p}{2}]
\end{align*}
In the last step, we use the symmetry of the potential and independence of the integrals in each group $G_l$. This independence by itself follows from the independence of the random walks $x_i(s)$, $i=1,2,\cdots,p$. 

But $E_{(0,0)}[\exp(p\int_0^tB(x_1(s)-x_2(s))ds]=y(t,x_1,x_2)$ satisfies two-particle Schr\"{o}dinger equation 
\begin{align}
    \frac{\partial y}{\partial t}&=\varkappa(\mathcal{L}_{x_1}+\mathcal{L}_{x_2})y+pB(x_1-x_2)y\\
    y(0,x_1,x_2)&=1.
\end{align}
Since we have a homogeneous space, $y(t,x_1,x_2)$ is a function depending on $x_1-x_2$, that is $y(t,x_1,x_2)=\tilde{y}(t,x_1-x_2)$. After a change of variable $v=x_1-x_2$ and considering initial conditions $y(t,x_1,x_2)=\tilde{y}(0,x_1-x_2)=1$, we eventually reduce the analysis to the problem:
\begin{align*}
\frac{\partial \tilde{y}}{\partial t}=\tilde{H}_2\tilde{y}=2\varkappa\mathcal{L}_{v}{\tilde{y}}(v)+pB(v)\tilde{y},\,\, \tilde{M}(0,\cdot)=1, \,\,v \in Z^d
\end{align*}
where the potential $B(\cdot)\in l^1(Z^d)$ is positive definite, as discussed in the introduction.  Let $\tilde{\gamma}_2(\varkappa,p)$ be the Lyapunov exponent for $\tilde{y}$. Due to Lemma (\ref{lem:two gamma relationship}),
\begin{align}
    \tilde{\gamma}_2(\varkappa,p)=p\tilde{\gamma}_2(\frac{\varkappa}{p},1)=p\gamma_2(\frac{\varkappa}{p}).
\end{align}
Combining together all results, we get for arbitrary $p\geq 3$ and $l=[\frac{p}{2}]$
\begin{align*}
    \lim\limits_{t\rightarrow \infty}\frac{\ln M_p(t)}{t}=\gamma_p\leq lp\gamma_2(\frac{\varkappa}{p})
\end{align*}
That is,
\begin{align}
    \frac{\gamma_p}{p}\leq l\gamma_2(\frac{\varkappa}{p}=[\frac{p}{2}]\gamma_2(\frac{\varkappa}{p}).
\end{align}
\end{proof}
The positivity $\gamma_2(\varkappa)>0$ will be discussed in section 3.

\subsection{Lower estimate for Lyapunov exponent $\gamma_p$}

\begin{theorem}
    For $p\geq 2$, 
    \begin{align}
        \frac{\gamma_p}{p-1}>\gamma_2(\frac{\varkappa}{p-1})+\frac{B(0)(p-2)}{2}-2\varkappa d
    \end{align}
\end{theorem}

\begin{proof}
Due to Kac-Feynman formula,
\begin{align*}
M_p(t)&=E_{(0,\cdots,0)}[\exp(\sum\limits_{1\leq i<j\leq p}(\int_0^t B(x_i(s)-x_j(s))ds))]\\
&\geq  E_{(0,\cdots,0)}[\exp(\sum\limits_{1\leq i<j\leq p}(\int_0^t B(x_i(s)-x_j(s))ds)) I_{x_l(s)=0, l=2,3,\cdots,p}],\,\,\\
&=  E_{(0,\cdots,0)}\left[\exp((p-1)(\int_0^t B(x_1(s)ds))\right]e^{\frac{(p-1)(p-2)}{2}B(0)}\prod\limits_{l=2}^pP(x_l(s)=0|x_l(0)=0).
\end{align*}
Hence, 
\begin{align*}
\ln M_p(t)&\geq \ln  E_{(0,\cdots,0)}[\exp((p-1)(\int_0^t B(x_1(s)ds))]\\+&\frac{(p-1)(p-2)}{2}B(0)t+\sum\limits_{l=2}^p \ln P(x_l(s)=0|x_l(0)=0).
\end{align*}
Since the number of the jumps of the random walk $x_l(s),s\in [0,t]$ is a Poisson process with the rate $2d\varkappa$, and the event $\{x_l(s)=0\}$ is equivalent to the absence of the jumps during $[0,t]$. Thus $P(x_l(s)=0|x_l(0)=0)=e^{-2d\varkappa t}$. And since for $l=2,3,\cdots, p$, $x_l(s)$ are independent random walks following the same law. Thus 
\begin{align}
\frac{\ln M_p(t)}{t(p-1)}&\geq \frac{\ln  E_{(0,\cdots,0)}[\exp((p-1)(\int_0^t B(x_1(s)ds))]}{t(p-1)}+\frac{(p-2)}{2}B(0)-2d\varkappa.
\end{align}

As $t\rightarrow \infty$,
\begin{align}
    \frac{\gamma_p}{p-1}\geq \frac{1}{p-1}\tilde{\gamma}_2(\varkappa,p-1)+\frac{(p-2)}{2}B(0)-2d\varkappa.
\end{align}
Due to Lemma \ref{lem:two gamma relationship}, 
\begin{align}
    \frac{\gamma_p}{p-1}\geq \frac{1}{p-1}\tilde{\gamma}_2(\varkappa,p-1)+\frac{(p-2)}{2}B(0)-2d\varkappa=\gamma_2(\frac{\varkappa}{p-1})+\frac{p-2}{2}B(0)-2d\varkappa
\end{align}

\end{proof}

\section{Spectral Analysis of $H_2$ and the strong intermittency}\label{sec:spectral analysis} 

The analysis of the second Lyapunov exponent $\gamma_2$ (or the top of the spectrum of the operator $H_2$) depends significantly on the properties of the random walk $x(t)$ associated with ``Laplacian" $\varkappa\mathcal{L}$. If $p(t,x,y)=P\{x(t)=y|x(0)=x\}$, $t>0, x,y\in Z^d$, then $p(t,x,y)$ is the fundamental solution of the following parabolic problem:
\begin{align}\label{p(t,x,y)equation}
\frac{dp}{dt}&=2\varkappa\mathcal{L} p, t>0, \quad
p(0,x,y)=\delta_y(x).
\end{align}

The process $x(t)$ is called recurrent if $\int_{0}^\infty p(t,x,x)dt=\infty$, otherwise, it is transient. 

 Using the Fourier transform, we obtain:
\begin{equation}
p(t,x,y)=\frac{1}{(2\pi)^d}\int_{T^d}e^{2\varkappa \hat{\mathcal{L}}(k)t}e^{i k (x-y)} dk, \,\,\, k\in T^d=[-\pi,\pi]^d.
\end{equation}
Here, $\hat{\mathcal{L}}(k)=\sum\limits_{z\neq 0}(\cos(k,z)-1)a(z)=\hat{a}(k)-1$ represents the Fourier symbol of the operator $\mathcal{L}$.
When $x=y$, 
\begin{align*}
\int_0^{\infty}p(t,x,x)dt=\frac{1}{(2\pi)^d}\int_{T^d}\frac{dk}{2\varkappa(1-\hat{a}(k))}.
\end{align*}

 If $E[x^2(t)]<\infty$, i.e., $\sum\limits_{z\in Z^d}|z|^2a(z)<\infty$, then $\hat{a}(k)=1-\frac{1}{2}(Dk,k)+O(|k|^2)$, where $D$ is the matrix of the second moments of $x(t)$. In this case $p(t,x,x)\sim \frac{c}{t^{d/2}},~t\to\infty,$, and the process $x(t)$ is transient for dimensions $d\geq3$ and recurrent for $d=1,2$.   
 
We consider a more general situation when the second moment of  $x(t)$ is infinite. We impose a certain regularity condition on the symbol $-\hat{\mathcal{L}}(k)=1-\hat{a}(k)$ of the operator $\mathcal{L}$. We assume that 
\begin{equation}\label{asm}
\hat{a}(k)=e^{-\beta(\dot{k})|k|^{\alpha}(1+o(|k|))}, \quad |k|\rightarrow 0,
\end{equation} 
 where $\dot{k}=\frac{k}{|k|}$, $0<\alpha\leq 2$, and $\beta(\dot{k})=\beta(-\dot{k})>0$ is a symmetric and positive function on $T^d$. Both functions $\beta$ on the sphere and the remainder term in the exponent are assumed to be sufficiently smooth. Corresponding conditions on the jump distribution $a(x)$ can be found in \cite{agbor2015global}.

Under these conditions, $p(t,x,x)\sim t^{-d/\alpha}$ as $t\to\infty$. Since $p(t,x,y)\leq 1$, we have $p(t,x,x)\leq Ct^{-d/\alpha}, ~t\geq 0,$ and the process $x(t)$ is recurrent for $d=2,\alpha=2$, and for $d=1,1\leq \alpha\leq 2$. It is transient in the case of $d=2,\alpha<2$ and $d=1,\alpha<1$.

\subsection{Transient Case}
This section is devoted to the question of the existence (or non-existence) of a positive eigenvalue of Schrödinger operators 
\begin{equation}\label{ode1A}
H\psi=\mathcal{L}\psi(x)+\sigma V(x)\psi(x), \quad x\in\mathbb Z^d, \quad \sigma>0,
\end{equation}
that includes the spectral problem for the operator  $H_2$ defined in (\ref{2ndmoment_L}) as a particular case if we set $\sigma=\frac{1}{2\varkappa}$. The corresponding change of the parameter appears after the rescaling $t\to t/(2\varkappa)$ in (\ref{2ndmoment_L}). We do not assume that the potential is positive definite and use notation $V$ instead of $B$.

Assume that potential $V(x)$ satisfies the Cwickel-Lieb-Rozenblum (CLR) type condition \cite{molchanov2010general},
\begin{align}\label{condition:ZLR1}
    \sum\limits_{x\in Z^d}|V(x)|^{\frac{d}{\alpha}}<\infty,\quad d=2,~\alpha<2 ~~ {\rm or}~~d=1,~\alpha<1,
\end{align}

or 
\begin{align}\label{condition:ZLR2}
    \sum\limits_{x\in Z^d}|V(x)|^{\frac{d}{2}}<\infty, \quad d\geq 3.
\end{align}
Let us note that these conditions hold the potential $B$ from section 2 since it was assumed that $ B\in L^1(\mathbb{Z^d})$.  

Denote the number of positive eigenvalues of the operator (\ref{ode1A}) by $N_0(\sigma V)=\sharp\{\lambda_i:\lambda_i(H)>0\}$.

\begin{theorem}
Let conditions (\ref{asm}), (\ref{condition:ZLR1}) and (\ref{condition:ZLR2}) hold and let the potential $V(\cdot)$ be positive at least at one point. Then there exists a critical value $\sigma_{cr}$ such that there are no positive eigenvalues when $\sigma<\sigma_{cr}$ and the positive eigenvalues $\lambda_0(H)$ exist when $\sigma>\sigma_{cr}$. 
\end{theorem}

\begin{proof}
  For arbitrary $\sigma$, the following Bargmann type estimate can be found in \cite{molchanov2012bargmann,reed1972methods}:
\begin{align*}\label{ineq:ZLR estimate 1}
N_0(\sigma V)\leq C\sum\limits_{x\in Z^d}\sigma V(x)\int_{\frac{1}{\sigma|V(x)|}}^{\infty} p(t,x,x)dt.
\end{align*}
Since $p(t,x,x)\leq Ct^{-d/\alpha},~ t\geq 0$, and $d>\alpha$, we obtain that
\begin{align*}
N_0(\sigma V)\leq\frac{C\alpha}{d-\alpha}\sigma^{\frac{d}{\alpha}}\sum\limits_{x\in Z^d}V(x)^{d/\alpha}.
\end{align*}
 Thus, $N_0(\sigma V)<1$ when $\sigma$ is small, i.e., there are no positive eigenvalues for small $\sigma$.

In order to prove existence of positive eigenvalues for large $\sigma$, consider a test function  $\psi_0(x)=\delta_{x_0}(x),$ where $x_0$ is a point such that $V(x_0)>0$. Then
\[
(H\psi_0,\psi_0)=l+\sigma V(x_0), \quad l=(\mathcal{L}\psi_0,\psi_0)<\infty.
\]
The quadratic form is positive for large $
\sigma,$ and therefore the positive eigenvalues exist.

\end{proof}

If in any dimension $d\geq 1$, the potential is positive and decays slowly: $V(x)\geq \frac{C}{1+|x|^{\frac{d}{\alpha}-\varepsilon}}$, then $H$ with arbitrary $\sigma >0$ has infinitely many positive eigenvalues. The proof is similar to the one in \cite{molchanov2022positive}.

\subsection{Recurrent Case} We continue to study the question of the existence of a positive eigenvalue for the Schrödinger operator
\begin{equation}\label{ode1a}
 H\psi=\mathcal{L}\psi(x)+\sigma V(x)\psi(x), \quad x\in\mathbb Z^d, \quad \sigma>0; \quad \sum_{x\in Z^d}|V(x)|<\infty,  \quad \sum_{x\in Z^d}V(x)\geq 0,
\end{equation}
but now we assume that operator $\mathcal{L}$ is recurrent and the conditions on $V$ are different compared to the transient case.

Since an analysis of the ground states 
is important for many applications and has been widely studied, we omit many restrictions imposed on $H$ earlier.
Operator $\mathcal{L}$ now can be not related to a Markov process and is defined simply as an operator which  in the Fourier images acts as the multiplication by a continuous function $\hat{\mathcal{L}}(k), ~k\in T^d$:
\begin{equation*}
\sum_{x\in Z^d}\mathcal{L}\psi(x)e^{-i<k,x>}=\hat{\mathcal{L}}(k)\sum_{x\in Z^d}\psi(x)e^{-i<k,x>}.  
\end{equation*}
We assume that $\hat{\mathcal{L}}(0)=0, ~~\hat{\mathcal{L}}(k)<0$ for $k\neq 0$. If the opposite is not stated, we assume that
\begin{equation}\label{it}
\int_{T^d}[\hat{\mathcal{L}}(k)]^{-1}dk=\infty.
\end{equation}
This condition in the Markovian case is equivalent to the recurrency of the process. While the last condition in (\ref{ode1a}) often appears as a consequence of the potential being positive-definite, we do not assume the latter property.

The spectrum of operator $\mathcal{L}$ coincides with the range of function $[\min_k\hat{\mathcal{L}}(k), 0]$. The presence of the potential may lead to the appearance of positive eigenvalues. Our goal is to find conditions that guarantee the existence of a positive eigenvalue for the operator (\ref{ode1a}) with arbitrary small $\sigma>0$. Note that the same conditions will lead to the existence of a positive eigenvalue for all $\sigma>0$. Indeed, if a positive eigenvalue exists for some $\sigma>0$  then it exists for operator $\sigma^{-1}\mathcal{L}+ V(x)$ 
and then it exists for the same operator with larger values of $\sigma$ due to the monotonicity 
of the latter operator in $\sigma$. Hence, operator (\ref{ode1a}) with larger values of $\sigma$ also has positive eigenvalues.

First, we will prove a couple of statements concerning the case when
\begin{equation}\label{posit}
\sum_{x\in Z^d}V(x)>0.
\end{equation}
\begin{theorem}\label{t}
Let (\ref{it}), (\ref{posit}) hold. Then the positive spectrum of the operator (\ref{ode1a}) is discrete and non-empty, i.e., the operator with arbitrary $\sigma>0$ has a positive eigenvalue.
\end{theorem}
\begin{remark}
This statement in the continuous case with a slightly more restrictive assumption was proven in \cite{molchanov2022positive}.
\end{remark}
\begin{proof}
Let us show that the positive spectrum of the operator (\ref{ode1a}) is non-empty.  Conditions on $V$ imply that $\widehat{V}(0)>0$ and $\widehat{V}(k)$ is continuous. Thus, there exists $\varepsilon_1>0$ such that $\widehat{V}(k-\xi)>v_0>0$ when $|k|,|\xi|<\varepsilon_1$.
We construct test functions $\psi(x)=\psi_\varepsilon(x)$ using their Fourier images:
\[
\widehat{\psi}_\varepsilon(k)=\left\{
                            \begin{array}{c}
                              [\hat{\mathcal{L}}(k)]^{-1},~~ 0<\varepsilon\leq |k|\leq \varepsilon_1, \\
                              0,~~ |k|\notin(\varepsilon,\varepsilon_1).  \\
                            \end{array}  \right.
\]
Then
\[
<H\psi,\psi>=\int_{\varepsilon<|k|<\varepsilon_1}\hat{\mathcal{L}}(k)|\widehat{\psi}(k)|^2dk+
\sigma\int_{\varepsilon<|k|<\varepsilon_1}\int_{\varepsilon<|k|<\varepsilon_1}\widehat{V}(k-\xi)\widehat{\psi}(k)\overline{\widehat{\psi}(\xi)}dkd\xi
\]
\begin{equation}\label{qf}
\leq\int_{\varepsilon<|k|<\varepsilon_1}[\hat{\mathcal{L}}(k)]^{-1}dk+\sigma v_0
\left(\int_{\varepsilon<|k|<\varepsilon_1}[\hat{\mathcal{L}}(k)]^{-1}dk\right)^2\to\infty \quad {\rm as} ~~\varepsilon\to 0,
\end{equation}
due to (\ref{it}). Hence, there are test functions for which the quadratic form (\ref{qf}) is positive, and therefore, the positive spectrum of $H$ is non-empty.

It remains to show that the positive spectrum of $H$ is discrete. The operators $(\mathcal{L}-\lambda)^{-1}$ in $L^2(Z^d)$ is analitic in $\lambda$ when $\lambda\notin [0,\infty)$, and the same is true for $T_\lambda:=V(\mathcal{L}-\lambda)^{-1}$ since $|V(x)|<C<\infty, ~x\in Z^d$. Hence the resolvent identity and the analytic Fredholm theorem will imply the discreteness of the positive spectrum of $H$ if we show that operators  $T_\lambda, ~\lambda\notin [0,\infty),$  are compact.   We represent $T_\lambda$ as the sum $T_\lambda=\alpha(x/n)T_\lambda+(1-\alpha(x/n))T_\lambda$ where $\alpha(x)=1$ when $|x|\leq 1, ~\alpha (x)=0$ when $|x|> 1$. The operator $\alpha(x/n)T_\lambda$ is finite-dimensional, and therefore it is compact, and
\[
\|(1-\alpha(x/n))T_\lambda\|_{L^2(Z^d)}\to 0   \quad {\rm as} \quad  n\to\infty
\]
since $V$ vanishes at infinity. Thus operators $T_\lambda, ~\lambda\notin [0,\infty),$  are compact as limits of compact operators $\alpha(x/n)T_\lambda$ as $n\to\infty$.

The proof is complete
\end{proof}
\begin{theorem}
Let condition (\ref{posit}) hold. Let operator $\mathcal{L}$ be a generator of a random walk and
\begin{equation}\label{con}
\int_{T^d}|k|[\hat{\mathcal{L}}(k)]^{-1}dk<\infty  \quad {\rm and} \quad \sum _{Z^d}|xV(x)|<\infty.
\end{equation}
Then operator $H$ defined in (\ref{ode1a}) with
\[
0<\sigma <\sigma_0=(2\int_{T^d}|k|[\hat{\mathcal{L}}(k)]^{-1}dk \sum _{Z^d}|xV(x)|)^{-1}
\]
has exactly one positive eigenvalue.
\end{theorem}
\begin{proof}
The following estimate for the number $N_0(V)$ of positive eigenvalues of operator $H$ follows from \cite{molchanov2012bargmann}, see Theorem 1.2:
\[
N_0(V)\leq1+2\sum_{x\in Z^d}|V(x)|\lim_{\lambda\to+0}[R_\lambda(x)-R_\lambda(0)]|,
\]
where  $R_\lambda(x-y)$ is the kernel of  the resolvent of operator $\mathcal{L}$.
The limit above does not exceed
\[
\int_{T^d}|e^{ikx}-1|[\hat{\mathcal{L}}(k)]^{-1}dk\leq\int_{T^d}|kx|[\hat{\mathcal{L}}(k)]^{-1}dk.
\]
The last two relations imply that $N_0(V)<2$ when  $\sigma<\sigma_0$. Thus $H$ has at most one eigenvalue in this case, and it exists due to Theorem \ref{t}.
\end{proof}
\subsection {Zero average potentials, one-dimensional case }
The goal of this subsection is to expand the important part of Theorem \ref{t} on the existence of positive eigenvalues to a particular case of operators $H$ with $\sum_{x\in Z^d}V(x)\geq0$ where the sum can be equal to zero. The operator $H$ will be a perturbation of the second derivative on the one-dimensional lattice. We would like to get this result also in the continuous case and therefore, we will consider simultaneously operators $H$ in $L^2(\mathbb R)$ and in $L^2(\mathbb Z)$ defined by the relations
\begin{equation}\label{ode1}
H\psi=\psi''(x)+\sigma V(x)\psi(x), \quad x\in\mathbb R ~{\rm or}~ \mathbb Z,  \quad {\rm where}    \quad \int_{-\infty}^\infty (1+|x|)|V(x)|dx<\infty, ~~V\not\equiv 0,
\end{equation}
Here and below, $\psi''(x)$ and the integrals must be replaced by $\psi(x-1)+\psi(x+1)-2\psi(x)$ and the corresponding infinite sums, respectively, when $ x\in\mathbb Z$. Let us stress that the assumption on the behavior of the potential at infinity in (\ref{ode1}) is stronger than in (\ref{ode1a}).
\begin{theorem}\label{tl}
If $\int_{-\infty}^\infty V(x)dx\geq0,~~V\not\equiv 0,$ then operator $H$ defined in (\ref{ode1}) has a positive eigenvalue $\lambda=\lambda(\sigma)$ for each $\sigma>0$.
\end{theorem}
\begin{remark} The proof below will be given in the continuous case $x\in\mathbb R$, but one needs only to replace integrals with sums for the proof to be valid in the lattice case.
\end{remark}\begin{remark} This statement for the operator on the half-axis $x>0$ with the boundary condition $\psi'(0)=0$ was proved in \cite{molchanov2022positive}
\end{remark}
We will need the following lemma (a similar statement can be found in \cite{molchanov2022positive}).
\begin{lemma}\label{ml}
When $|\sigma|\ll1$, there are solutions $\psi_1,\psi_2$ of the equation $H\psi=0$ on the half axis $x\geq0$, where $H$ is defined in   (\ref{ode1}), such that
\begin{equation}\label{as}
\psi_1(x)=1+O(\sigma), \quad \psi_1'(x)=O(\sigma), \quad x\geq0;~ \quad  \psi_1(0)=1,\quad \psi_1'(0)=\alpha(\sigma),
\end{equation}
\begin{equation}\label{as1}
\psi_2(x)=x(1+O(\sigma)), \quad \psi_2'(x)=1+O(\sigma), \quad x\geq0;~\quad  \psi_2(0)=0,\quad \psi_2'(0)=1.
\end{equation}
Here the remainder terms are uniform in $x$ and
\begin{equation}\label{alp}
\alpha(\sigma)=\sigma b+\sigma^2c^2+O(\sigma^3),  \quad b=\int_0^\infty V(\xi)d\xi, \quad c^2=\int_0^\infty \left(\int_\eta^\infty V(\xi)d\xi\right)^2d\eta.
\end{equation}
\end{lemma}
\begin{proof}
We look for a complementary bounded solution $\psi$ in the form $\psi=1+z(x)$. Then $z''+\sigma V(x)z=-\sigma V(x)$, and this equation can be reduced to the integral equation
\begin{equation}\label{eq1}
z(x)+\sigma Pz=-\sigma\int_x^\infty (\xi-x)V(\xi) d\xi, \quad x\geq 0, \quad {\rm where}    \quad Pz=\int_x ^\infty(\xi-x)V(\xi)z(\xi)d\xi.
\end{equation}
 Operator $P$ in the space $C$ of continuous bounded functions on the semi-axis $x\geq0$ is bounded, and the right-hand side in (\ref{eq1}) has the from $ -\sigma f$, where $f\in C$. Thus equation (\ref{eq1}) is uniquely solvable in $C$ when $|\sigma|\ll1$ and the solution is analytic in $\sigma$. Hence, 
 \begin{equation}\label{eq15}
 \psi=1-\sigma f+\sigma^2Pf+O(\sigma^3),
\end{equation}
where $-\sigma f$ is the right-hand side in (\ref{eq1}). Since (\ref{eq1}) also admits the differentiation in $x$, the latter expansion can be differentiated in $x$. This implies that $\psi$ satisfies the first two relations in (\ref{as}) and has the following initial data:
\[
\psi(0)=1-\sigma a+O(\sigma^2),\quad \psi'(0)=\sigma b+\sigma^2(-ab+c^2)+O(\sigma^3),
\]
where $a=\int_0^\infty \xi V(\xi) d\xi, $ and $b,~c^2$ are  defined in (\ref{alp}).
The evaluation of the coefficient for $\sigma^2$ in the expansion of $\psi'(0)$ is based on (\ref{eq15}) but it also involves an integration by parts, and more details can be found in \cite{molchanov2022positive} (see Lemma 3.3) if needed. 

Function $\psi_1$ can be defined as $\psi_1(x)=\psi(x)/\psi(0)$.
Function $\psi_2$ can be constructed as
\begin{equation*}
\psi_2(x)=\psi_1(x)\int_{0}^x \psi_1^{-2}(\xi)d\xi, \quad x\geq 0, \quad |\sigma|\ll1.
\end{equation*}
Then (\ref{as1}) follows from (\ref{as}).

\end{proof}

{\it Proof of Theorem \ref{tl}.} Theorem 3.1 from \cite{molchanov2012bargmann} implies that the positive spectrum of $H$ is discrete. Thus it is enough to show that the positive spectrum is not empty. It is enough to show the existence of a positive spectrum of $H$ for arbitrarily small $\sigma>0$ since it will imply its existence for all $\sigma>0$, see the paragraph before the statement of the theorem. Hence, the theorem will be proven if, for some sequence $\sigma=\sigma_n\to0, ~n\to\infty,$  we construct functions $\phi=\phi_n(x)\in C^1(\mathbb R)$ with compact supports such that
$$
\int_{-\infty}^\infty[-\phi'^2(x)+\sigma_n V(x)\phi^2(x)]dx>0.
$$

We define solutions $\psi_3(x), \psi_4(x)$ of equation (\ref{ode1}) on the negative half-axis $x\leq0$ similar to $\psi_1(x), \psi_2(x)$ constructed in Lemma \ref{ml} on the semi-axis $x\geq0$. In particular,
\begin{equation}\label{bet}
\psi_3(0)=1,\quad \psi_3'(0)=\beta(\sigma),\quad \psi_4(0)=0,\quad \psi_4'(0)=1,
\end{equation}
where
\[
\beta=-\sigma b_--\sigma^2(c_-)^2+O(\sigma^3), \quad b_-=\int_{-\infty}^ 0V(\xi)d\xi, \quad (c_-)^2=\int_{-\infty}^0 \left(\int_{-\infty}^\eta V(\xi)d\xi\right)^2d\eta.
\]
For example, one can find $\psi_1(x), \psi_2(x)$ on $x\geq0$ for the potential $V(-x)$ and put$$
\psi_3(x)= \psi_1(-x),\quad \psi_4(x)= \psi_2(-x), \quad x\leq0.
$$

Let $m=[\sigma^{-3}]$ be the integer part of $\sigma^{-3}$, where $0\leq\sigma\ll1$. 
%It can be seen from the proof that $m$ can be a large enough positive integer of order $O(\sigma^{-1})$ in some cases and of order  $O(\sigma^{-2})$ in other cases. 
We put
\begin{equation}\label{phi}
\phi=\left\{
                      \begin{array}{cc}
                        \psi_1(x)-\varepsilon_1\psi_2(x), &0\leq x\leq m, \\
                         \psi_3(x)+\varepsilon_2\psi_4(x), &-m\leq x\leq0, \\
                        0, & ~~~|x|>m,\\
                      \end{array}
                    \right.
\end{equation}
where $\varepsilon_{1,2}$ are positive solutions of the equations $\psi_1(m)-\varepsilon_1\psi_2(m)=0, ~\psi_3(-m)+\varepsilon_2\psi_4(-m)=0$, respectively. From the asymptotics of $\psi_i(x),~1\leq i\leq4,$ at infinity, it follows that solutions $\varepsilon_{1,2}$ exist and $\varepsilon_{1,2}=m^{-1}(1+O(\sigma))=O(\sigma^{-3})$. Note that $\phi(\pm m)=0$.

Since $ H\phi(x)=0, ~0<|x|<m,$ and $\phi(x)$ is continuous with $\phi(0)=1$, we have
$$
H\phi=\varepsilon \delta(x), \quad |x|<m, \quad \varepsilon=\phi'(+0)-\phi'(-0),
$$
where $\delta(x)$ is the delta function, and therefore,
$$
\int_{|x|<m}(H\phi(x))\phi(x)dx=\varepsilon.
$$
On the other hand, integration by parts implies
$$
\int_{|x|<m}(H\phi(x))\phi(x)dx=\int_{|x|<m}[-\phi'^2(x)+\sigma V(x)\phi^2(x)]dx-\varepsilon.
$$
Thus
\begin{equation}\label{fo}
\int_{\mathbb R}([\phi'^2(x)-\sigma V(x)\phi^2(x)]dx=2\varepsilon.
\end{equation}

Asymptotics of $\psi_i(x), ~1\leq i\leq 4$ as $x\to\pm 0$, see (\ref{as}), (\ref{as1}), (\ref{bet}), lead to
\[
\varepsilon=\sigma\int _\mathbb R V(x)dx+\sigma^2(c^2+(c_-)^2)+O(\sigma^3)-\varepsilon_1-\varepsilon_2
\]
\[
=\sigma\int _\mathbb R V(x)dx+\sigma^2(c^2+(c_-)^2)+O(\sigma^3).
\]
The coefficient for $\sigma$ here is not negative, and the coefficient for $\sigma^2$ is positive if $V\not\equiv0$. Hence, $\varepsilon>0$ when $\sigma>0$ is small enough. Thus, for that $ \sigma$,  the quadratic form in the left-hand side of (\ref{fo}) is positive, and therefore the positive spectrum of $H$ is not empty.

\qed

%    Bibliographies can be prepared with BibTeX using amsplain,
%    amsalpha, or (for "historical" overviews) natbib style.
\bibliographystyle{amsplain}
\bibliography{mybib}
%    Insert the bibliography data here.

\end{document}